\documentclass{article}

\usepackage{microtype}
\usepackage{graphicx}
\usepackage{subfigure}
\usepackage{booktabs} 
\usepackage{hyperref}

\usepackage[accepted]{formatting/mlsys2020}

\usepackage{tikz}
\usepackage{xcolor}
\usepackage{xspace}
\usepackage{xurl}
\usepackage{float}
\usepackage{amsmath}
\usepackage{amssymb}
\usepackage{amsthm}
\usepackage{amsfonts}
\usepackage{pifont}
\usepackage{comment}
\usepackage{graphicx}
\usepackage[scaled=0.86]{helvet}

\newtheorem{theorem}{\normalfont\textsc{Theorem}}
\newcommand{\name}[0]{\textsc{Serenity}\xspace}

\newcommand{\PConvPeakImprovement}[0]{\textsf{1.86$\times$}}
\newcommand{\DPPeakImprovement}[0]{\textsf{1.68$\times$}}
\newcommand{\PConvPeakExtraImprovement}[0]{\textsf{10.7\%}}

\newcommand{\PConvSchedulingLatency}[0]{\textsf{48.8 secs}\xspace}
\newcommand{\DPSchedulingLatency}[0]{\textsf{40.6 secs}\xspace}

\newcommand{\OffChipImprovement}[0]{\textsf{1.76$\times$}\xspace}

\usepackage{formatting/ahn}

\usepackage[subtle]{savetrees}

\mlsystitlerunning{Ordering Chaos: Memory-Aware Scheduling of Irregularly Wired Neural Networks for Edge Devices}

\begin{document}
\twocolumn[
\mlsystitle{Ordering Chaos: Memory-Aware Scheduling of \\ Irregularly Wired Neural Networks for Edge Devices}
\mlsyssetsymbol{equal}{*}
\mlsyssetsymbol{intern}{${\dagger}$}
\begin{mlsysauthorlist}
\mlsysauthor{Byung Hoon Ahn}{ucsd,intern}
\mlsysauthor{Jinwon Lee}{qti}
\mlsysauthor{Jamie Menjay Lin}{qti}
\mlsysauthor{Hsin-Pai Cheng}{duke,intern}
\mlsysauthor{Jilei Hou}{qti}
\mlsysauthor{Hadi Esmaeilzadeh}{ucsd}

\end{mlsysauthorlist}

\mlsysaffiliation{qti}{Qualcomm AI Research}
\mlsysaffiliation{ucsd}{University of California, San Diego}
\mlsysaffiliation{duke}{Duke University}

\mlsyscorrespondingauthor{Byung Hoon Ahn}{bhahn@eng.ucsd.edu}

\mlsyskeywords{Scheduling, Compiler, Memory Footprint, Irregularly Wired Neural Networks, Randomly Wired Neural Networks, Neural Architecture Search, Machine Learning, MLSys, MLSys, Edge, Dynamic Programming, Graph Rewriting}

\vskip 0.3in

\begin{abstract}
Recent advance on automating machine learning through Neural Architecture Search and Random Network Generators, has yielded networks that deliver higher accuracy given the same hardware resource constrains, e.g., memory capacity, bandwidth, number of functional units.
Many of these emergent networks; however, comprise of irregular wirings (connections) that complicate their execution by deviating from the conventional regular patterns of layer, node connectivity, and computation.
The irregularity leads to a new problem space where the schedule and order of nodes significantly affect the activation memory footprint during inference.
Concurrently, there is an increasing general demand to deploy neural models onto resource-constrained edge devices due to efficiency, connectivity, and privacy concerns.
To enable such a transition from cloud to edge for the irregularly wired neural networks, we set out to devise a compiler optimization that caps and minimizes the footprint to the limitations of the edge device.
%
{
This optimization is a search for the schedule of the nodes in an intractably large space of possible solutions.
We offer and leverage the insight that partial schedules leads to repeated subpaths for search and use the graph properties to generate a signature for these repetition.
These signatures enable the use of \emph{Dynamic Programming} as a basis for the optimization algorithm.
However, due to the sheer number of neurons and connections, the search space may remain prohibitively large.
As such, we devise an \emph{Adaptive Soft Budgeting} technique that during dynamic programming performs a light-weight meta-search to find the appropriate memory budget for pruning suboptimal paths.
%
%
Nonetheless, schedules from any scheduling algorithm, including ours, is still bound to the topology of the neural graph under compilation.
To alleviate this intrinsic restriction, we develop an \emph{Identity Graph Rewriting} scheme that leads to even lower memory footprint without changing the mathematical integrity of the neural network.
}
We evaluate our proposed algorithms and schemes using representative irregularly wired neural networks.
Compared to TensorFlow Lite, a widely used framework for edge devices, the proposed framework provides \PConvPeakImprovement reduction in memory footprint and \OffChipImprovement reduction in off-chip traffic with an average of less than one minute extra compilation time.
\end{abstract}
]

\printAffiliationsAndNotice{$^{\dagger}$Work done as intern at Qualcomm AI Research.}


\section{Introduction}
\label{sec:intro}
%
%
%
%
Growing body of work focuses on Automating Machine Learning (AutoML) using Neural Architecture Search (NAS)~\cite{nas:2016,adanet:2017,nasnet:2018,darts:2018,proxylessnas:2018,amoebanet:2018,swiftnet:2019} and now even, Random Network Generators~\cite{randomnet:2019,wirings:2019} which emit models with \emph{irregular} wirings, and shows that such \emph{irregularly wired neural networks} can significantly enhance classification performance.
These networks that deviate from \emph{regular topology} can even adapt to some of the constraints of the hardware (e.g., memory capacity, bandwidth, number of functional units), rendering themselves especially useful in targeting edge devices.
Therefore, lifting the regularity condition provides significant freedom for NAS and expands the search space~\cite{adanet:2017,swiftnet:2019,randomnet:2019}.

The general objective is to enable deployment of neural intelligence even on stringently constrained devices by trading off regular wiring of neurons for higher resource efficiency.
Importantly, pushing neural execution to edge is one way to address the growing concerns about privacy~\cite{privacy:2020} and enable their effective use where connectivity to cloud is restricted~\cite{dl_edge:2019}.
However, the new challenge arises regarding orchestrating execution of these irregularly wired neural networks on the edge devices as working memory footprint during execution frequently surpass the strict cap on the memory capacity of these devices.
The lack of multi-level memory hierarchy in these micro devices exacerbates the problem, because the network cannot even be executed if the footprint exceeds the capacity.
To that end, despite the significant potential of irregularly wired neural networks, their \emph{complicated execution pattern}, in contrast to previously \emph{streamlined execution of models with regular topology}, renders conventional frameworks futile in taking these networks to edge due to their large \emph{peak memory footprint}.
While peak memory footprint is largely dependent on scheduling of neurons, current deep learning compilers~\cite{tvm:2018,tc:2018} and frameworks~\cite{tensorflow:2016,pytorch:2017,caffe:2014} rely on basic topological ordering algorithms that are oblivious to peak memory footprint and instead focus on an orthogonal problem of tiling and kernel level optimization.
This paper is an initial step towards embedding peak memory footprint as first-grade constraint in deep learning schedulers to unleash the potential of the emergent irregularly wired neural networks.
As such, this paper makes the following contributions:

\textbf{(1) Memory-aware scheduling for irregularly wired neural networks.}
Scheduling for these networks is a topological ordering problem, which enumerates an intractably large space of possible schedules.
We offer and leverage the insight that partial schedules leads to repeated subpaths for search and use the graph properties to generate a signature for these repetition while embedding a notion of the running memory usage.
These signatures enable the use of \emph{Dynamic Programming} as a basis for the optimization algorithm.

\vspace{-1.3ex}
\textbf{(2) Adaptive soft budgeting for tractable compilation time.}
Even with the dynamic programming as the base, due to the sheer number of neurons and connections, the search space may remain too large (exponentially large) in practice.
As such, we devise an \emph{Adaptive Soft Budgeting} technique that uses a lightweight meta-search mechanism to find the appropriate memory budget for pruning the suboptimal paths.
This technique aims to find an inflection point beyond which tighter budgets may lead to no solution and looser budget prolongs the scheduling substantially, putting the optimization in a position of questionable utility.

\vspace{-1.3ex}
\textbf{(3) Identity graph rewriting for enabling higher potential in memory reduction.}
Any scheduling algorithm, including ours, is still bound to the topology of the neural graph under compilation.
To relax this intrinsic restriction, we devise an \emph{Identity Graph Rewriting} scheme that exchanges subgraphs leading to a lower memory footprint without altering the mathematical integrity of the neural network.
%

Results show that our adaptive scheduling algorithm improves peak memory footprint for irregularly wired neural networks by \DPPeakImprovement compared to TensorFlow Lite, the de facto framework for edge devices.
Our graph rewriting technique provides an opportunity to lower the peak memory footprint by an additional \PConvPeakExtraImprovement.
Furthermore, our framework can even bring about \OffChipImprovement reduction in off-chip traffic for devices with multi-level memory hierarchy, and even eliminate the traffic in some cases by confining the memory footprint below the on-chip memory capacity.
These gains come at average of less than one minute extra compilation time.

\section{Challenges and Our Approach}
\label{sec:challenges}

\begin{figure}[t!]
	\centering
	\vspace{-1ex} 
	\subfigure[RandWire]{\includegraphics[width=0.375\linewidth]{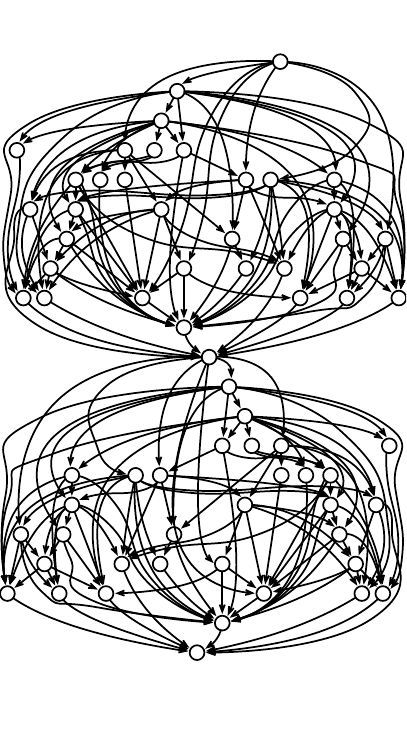}}
	\hspace{0.2in}
	\subfigure[SwiftNet]{\includegraphics[width=0.375\linewidth]{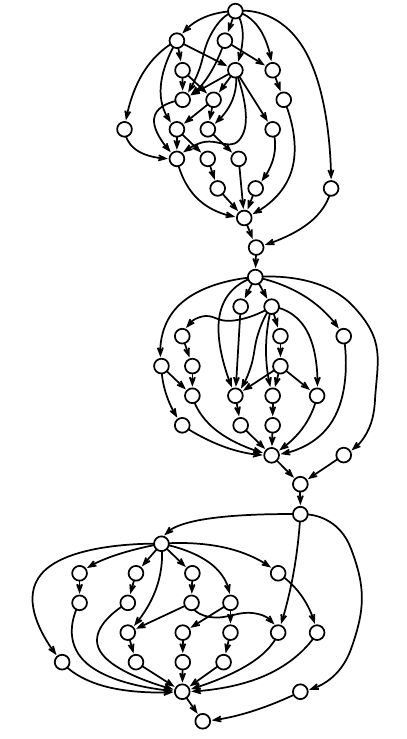}}
	\vspace{-2ex}
	\caption{Architecture of network models from NAS and Random Network Generators. Topology of such networks include distinctive \emph{irregular wirings} between the nodes.}
	\vspace{-2ex}
	\label{fig:complex_networks}
\end{figure}

\begin{figure}[b!]
	\centering
	\vspace{-2ex} 
	\includegraphics[width=0.9\linewidth]{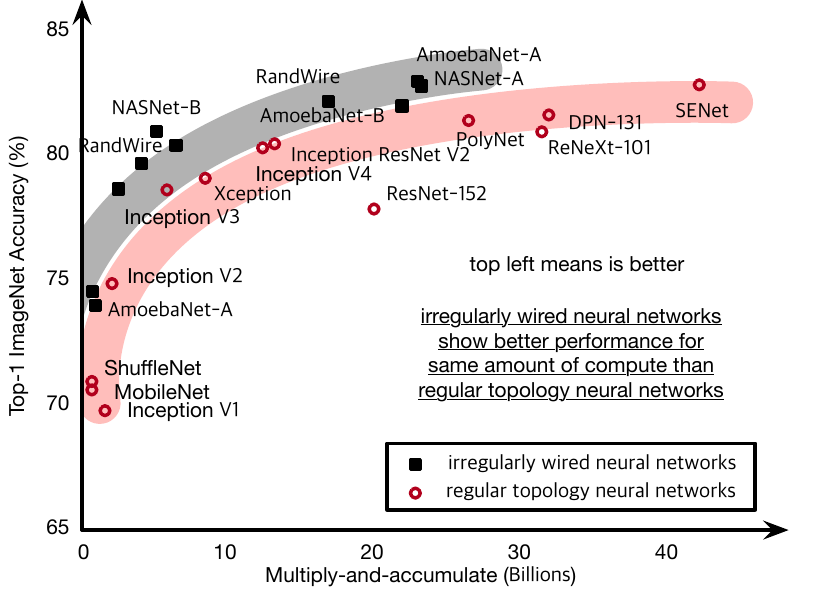}
	\vspace{-2ex}
	\caption{ImageNet accuracy vs number of multiply-and-accumulate, where irregularly wired neural networks show higher performance for same compute than regular topology neural networks. Plot for number of parameters also displays a similar trend.}
	\vspace{-2ex}
	\label{fig:network_performance}
\end{figure}

\subsection{Irregularly Wired Neural Networks}
Recent excitement in Automated Machine Learning (AutoML)~\cite{automl:2015,automl:2017,amc:2018,releq:2018,haq:2019,autokeras:2019} aims to achieve \emph{human out of the loop} in developing machine learning systems.
This includes Neural Architecture Search (NAS)~\cite{nas:2016,nasnet:2018,darts:2018,proxylessnas:2018,amoebanet:2018,swiftnet:2019} and Random Network Generators~\cite{randomnet:2019,wirings:2019} that focus on automation of designing neural architectures.
%
Figure~\ref{fig:complex_networks} demonstrates that networks of this regime are characterized by their distinctive \emph{irregular graph topology} with much more irregular wirings (dataflow) compared to conventional networks with regular graph topology.
This paper refers to these networks as \emph{irregularly wired neural networks}.

From the performance perspective, these networks have shown to outperform manually designed architectures in terms of accuracy while using less resources.
In fact, majority of winning neural architectures in competitions with primary goal of reducing resources~\cite{lpirc:2017} rely on NAS, suggesting its effectiveness in that respect.
Figure~\ref{fig:network_performance} plots the accuracy of different models given their computation.
The figure clearly shows that the Pareto frontier of irregularly wired neural networks from NAS and Random Network Generators are better than the hand designed models with regular topology.
This indicates that the efficiency in terms of accuracy given fixed resources are better with the irregularly wired neural networks.

\begin{figure}[t!]
	\centering
	\subfigure[SwiftNet Cell A.]{\includegraphics[width=0.45\linewidth]{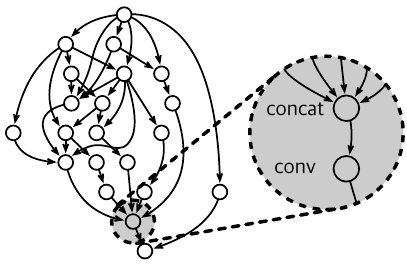}\label{fig:cell_a}}
	\hfill
	\subfigure[CDF of peak memory for different possible schedules.]{\includegraphics[width=0.45\linewidth]{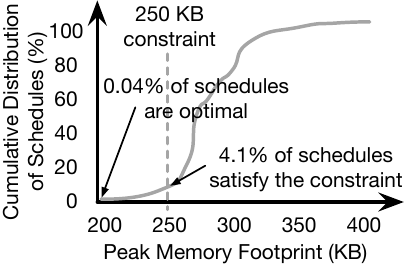}\label{fig:histogram}}
	\vspace{-2ex}
	\caption{CDF of the peak memory footprint for the different possible schedules of a given \emph{irregularly wired neural network}.}
	\vspace{-2ex}
	\label{fig:range}
\end{figure}

\subsection{Challenges}
Many existing compilers~\cite{tvm:2018,tc:2018} and frameworks~\cite{pytorch:2017,tensorflow:2016,caffe:2014} rely on basic topological ordering algorithms to schedule the graph.
%
While the current approach may be sufficient to run conventional networks on server-class machines, such scheme may be unfit for running irregularly wired neural networks on resource-constrained edge devices.
This is because, unlike running networks with regular topology, running irregular networks results in varied range of memory footprint depending on the schedule.
%
For instance, given the constraints of a representative edge device (SparkFun Edge: 250KB weight/activation memory and 60M MACs), Figure~\ref{fig:histogram} shows that 4.1\% of the schedules barely meets the hard memory constraint, while only 0.04\% would achieve the optimal peak memory.
%
In reality, such limitation will prevent further exploration regarding the diversity and innovation of network design, and in order to allow edge computing regime to take full advantage of the irregularly wired neural networks, this limitation should be alleviated if not removed.

\subsection{Design Objectives}
\label{sec:design_objective}
\paragraph{Scheduling algorithm.}
To address this issue, our work aims to find a schedule of nodes $s^*$ from the search space $\mathcal{S}$ that would minimize peak memory footprint $\mu_{peak}$.
$\mathcal{S}$ enumerates all possible orderings of the nodes $v \in \mathcal{V}$ where $\mathcal{V}$ is the set of all nodes within a graph $\mathcal{G}$.
%

\vspace{-2ex} 
\begin{equation}
s^* = \argmin_{s} \mu_{peak}(s, \mathcal{G}), \qquad \text{for $s \in \mathcal{S}$}
\label{eq:obj_scheduling}
\end{equation}

The most straightforward way to schedule is a \emph{brute force} approach which just enumerates $\mathcal{S}$ and picks one with the minimum peak memory footprint.
While this extreme method may find an optimal solution, it is too costly in terms of time due to its immense complexity: $\Theta(|V|!)$ where $|V|$ denotes number of nodes in the graph.
One way to improve is to narrow down the search space to just focus on only the \emph{topological orderings} $\mathcal{S}_T \subset \mathcal{S}$.
However, this will still suffer from a complexity with an upper bound of $\mathcal{O}(|V|!)$ (takes days to schedule DAG with merely 30 nodes).
In fact, previous works~\cite{codegen:1976,scheduling_complexity:1989} already prove optimal scheduling for DAGs is NP-complete.
On another extreme are heuristics for topological ordering such as Kahn's algorithm~\cite{kahn:1962}, with complexity of $\mathcal{O}(|V|+|E|)$ where $V$ and $E$ are number of nodes and edges.
However, as demonstrated in Figure~\ref{fig:range}, such method may yield suboptimal schedule of nodes which will not run on the target hardware.
To this end, we explore \emph{dynamic programming} combined with \emph{adaptive soft budgeting} for scheduling to achieve an optimal solution while keeping the graph constant $s^*$, without adding too much overhead in terms of time.
We explain our algorithms in depth in Section~\ref{sec:dp_scheduling} and ~\ref{sec:optimized_implementation}.

\paragraph{Graph rewriting.}
Any scheduling algorithm including ours is intrinsically bounded by the graph topology.
Therefore, we explore to transform the search space through \emph{graph rewriting}~\cite{graph_rewriting:1999}.
Graph rewriting is generally concerned with substituting a certain pattern in the graph with a different pattern to achieve a certain objective.
For a computational dataflow graph, leveraging \emph{distributive, associative, and commutative properties} within the computation of the graph, graph rewriting can maintain the semantics while bringing significant improvements regarding some objective.
%
%
For example, in general programs, $\sum_i log{x_i}$ can be represented as $\sum_{odd i} log{x_i} + \sum_{even i} log{x_i}$ or $log{\prod_i x_i}$, while $x + x$ can be translated to $x \times 2$ or $x << 1$.
Likewise, we bring this insight to neural networks to find a set of possible transformations $\mathcal{X}$ that can rewrite the original graph $\mathcal{G}$ to a new graph $\mathcal{G}'$ that would also change our search space $\mathcal{S}$ to one with a lower peak memory footprint:

\vspace{-2ex} 
\begin{equation}
\mathcal{X}^* = \argmin_{\mathcal{X}}(\mu_{peak}(s^*, \mathcal{X}(\mathcal{G})))
\label{eq:obj_rewriting}
\end{equation}

We identify a set of candidate patterns for transformation $\chi: g \rightarrow g'$ ($g \in \mathcal{G}$ and $g' \in \mathcal{G}'$), which constitutes $\mathcal{X}$.
While transforming the graph, our method keeps the \emph{mathematical integrity} of the graph intact, thus not an approximation method.
We embed this systematic way to improve peak memory footprint and the search space as \emph{identity graph rewriting}, and we address this technique in Section~\ref{sec:graph_rewriting}.

\begin{figure*}
	\centering
	\includegraphics[width=0.9\linewidth]{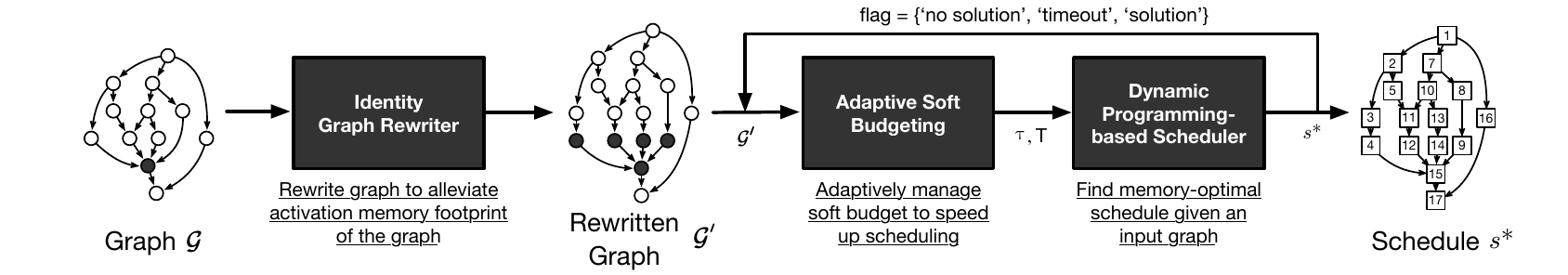}
	\vspace{-2ex}
	\caption{Overall workflow of \name, memory-aware scheduling of \emph{irregularly wired neural network}.}
	\vspace{-2ex}
	\label{fig:overview}
\end{figure*}

\section{\name: Memory-aware Scheduling of irregularly Wired Neural Networks}
\label{sec:methodology}

As discussed in Section~\ref{sec:challenges}, the objective is reducing the peak memory footprint while executing irregularly wired neural networks.
We propose \name, memory-aware scheduling that targets devices with restricted resources (e.g., edge devices).
Figure~\ref{fig:overview} summarizes the overall scheduling process, highlighting the major contributions of our approach.
Input to \name is a graph of irregularly wired neural network $\mathcal{G}$, which in fact acts as an intermediate representation (IR) during the scheduling process.
We augment this IR with the metadata of the nodes such as the operation type, input/output edges, input/output shapes, and memory cost.
Then the \emph{graph rewriter} transforms the graph $\mathcal{G} \rightarrow \mathcal{G}'$ to relax the memory costs of memory intensive patterns with the goal of reducing the peak memory footprint $\mu_{peak}$ of $\mathcal{G}$.
\name schedules the graph to an optimal schedule $s^*$ using the \emph{dynamic programming-based scheduler}.
However, since the scheduling may be slow due to the complexity, we scale down search space by leveraging \emph{divide-and-conquer} which partitions the graph into multiple subgraphs.
Them, we augment the scheduler with an \emph{adaptive soft budgeting} which prunes suboptimal paths by adaptively finding a budget for thresholding through a swift meta-search to speed up the scheduling process.
This section focuses on the innovations of \name: dynamic programming-based scheduling, divide-and-conquer, adaptive soft budgeting, and graph rewriting, which are explained in detail in Section~\ref{sec:dp_scheduling}, ~\ref{sec:optimized_implementation}, and ~\ref{sec:graph_rewriting}, respectively.

\subsection{Dynamic Programming-based Scheduling: Achieving Optimal Peak Memory Footprint}
\label{sec:dp_scheduling}

Our goal for the scheduling algorithm is to minimize the peak memory footprint $\mu_{peak}(s, \mathcal{G})$.
As stated in Section~\ref{sec:design_objective}, recursive algorithms that covers the entire search space $\mathcal{S}$ or the subspace of all topological orderings $\mathcal{S}_T \subset \mathcal{S}$ takes impractically long time.
This is primarily due to the repetitive re-computation of subproblems that upper bounds the algorithm by $\mathcal{O}(|V|!)$.
%
Therefore, we leverage \emph{dynamic programming}~
\cite{bellman:1961,dp:1966,held_karp:1962} which includes a \emph{memoization} scheme that has been shown to be effective in reducing the complexity of time-intensive algorithms by reusing solutions from their subproblems, while still finding optimal solution by sweeping the entire search space.

\begin{figure}[b!]
	\centering
	\vspace{-2ex} 
	\includegraphics[width=0.9\linewidth]{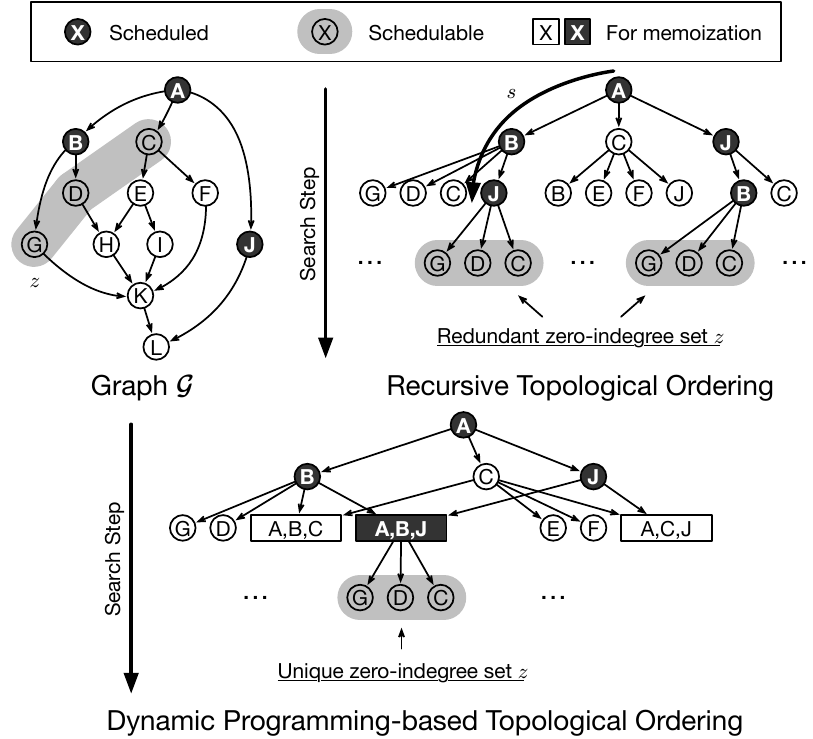}
	\vspace{-2ex}
	\caption{Illustration of identifying redundant \emph{zero-indegree} set $z$ and making $z$ unique (\emph{square}) throughout the topological ordering algorithm to reduce re-computation.}
	\vspace{-2ex}
	\label{fig:dp_subproblem}
\end{figure}

\paragraph{Identifying signature to enable dynamic programming.}
The first step to applying dynamic programming to a new problem is characterizing the structure of an optimal solution: $s^* = s_n^*$ ($s_n^*$ is an optimal solution for $n$ number of nodes).
Then, it requires identifying a recursive relationship between the optimal solution of a subproblem $s_{i}^*$ and the original problem $s_{i+1}^*$, and we do this by analyzing the straightforward \emph{recursive topological ordering}, which while inefficient sweeps the entire search space.
In essence, topological ordering algorithm is a repeated process of identifying a set of nodes that are available for scheduling and iterating the set for recursion.
In graph theory such a set of nodes available for scheduling is called \emph{zero-indegree set} $z$, where $z$ is a set of nodes which all of their incoming edges and the corresponding predecessor nodes (\emph{indegree}) have been scheduled.
Figure~\ref{fig:dp_subproblem} demonstrates the recursion tree of the different topological ordering algorithms, where the height of the tree is the search step and every path from the root to the leaf is a topological ordering $s \in S_T$.
The figure highlights the redundant $z$ in the recursive topological ordering in the recursion tree, then merges these $z$ to make them unique, identifying it as the signature for repetition, and prevent the aforementioned re-computation.
This makes the scheduling for $z$ into a unique subproblem, that constitutes the \emph{dynamic programming-based topological ordering}.

\begin{figure}[t!]
	\centering
	\includegraphics[width=0.9\linewidth]{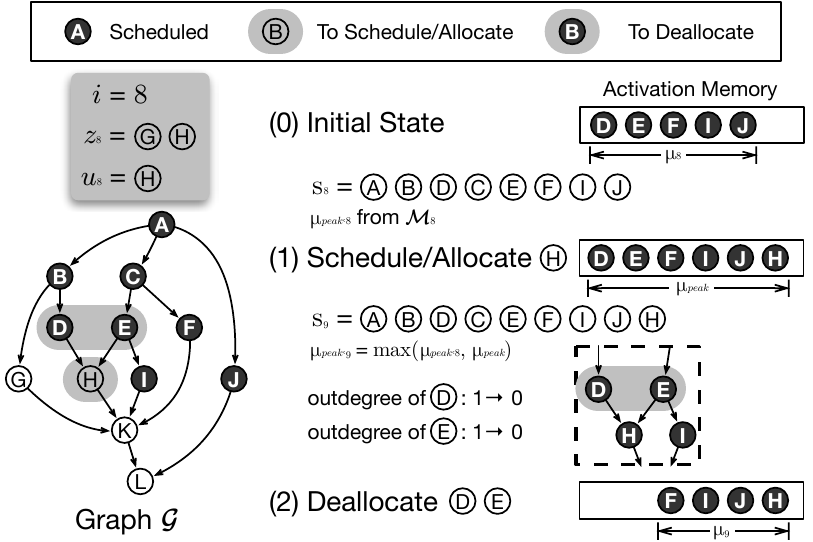}
	\vspace{-2ex}
	\caption{Visualization of scheduling the node $u_8$ = \protect{\circlewhite{\textsf{\size{8}{H}}}} during the search step $i=8$. Starting from $s_8$, $\mu_{8}$, and $\mu_{peak,8}$ the figure shows how the algorithm calculates $s_9$, $\mu_{9}$, and $\mu_{peak,9}$}
	\vspace{-2ex}
	\label{fig:dp_memory}
\end{figure}

\paragraph{Integrating the peak memory footprint constraint.}
On top of the dynamic programming formulation that shows potential for optimizing the search space significantly, we overlay the problem specific constraints to achieve the optimal solution.
In particular, we calculate the \emph{memory footprint} $\mu_{i+1}$ and its corresponding \emph{peak} $\mu_{peak,i+1}$ in each search step $i$ to select optimal path $s_{i+1}^*$ for \emph{memoization}.
Here, we clarify the process of a search step, explaining the details of calculating $\mu_{peak,i+1}$ and saving $s_{i+1}$ for each search step $i$.
In each search step, we start with number of \emph{unique} zero-indegree sets $z_i$ (signature), saved in $i^{th}$ entry of memoization $\mathcal{M}_i$.
For each $z_i$, we append the schedule up to the point $s_i$, sum of activations in the memory $\mu_i$ for the signature $z_i$, and the peak memory footprint of the $s_i$ denoted $\mu_{peak,i}$.
Therefore, in each search step $i$, we start with $s_i$, $\mu_i$, and $\mu_{peak,i}$ for $s_i$.
Then, when we iterate $z_i$ to schedule a new node $u_i$, its output activation is appended to $s_i$ to form $s_{i+1}$, and is \emph{allocated} in the memory.
Size of $u_i$ is product ($\prod$) of $u_i.$shape, where shape is a property of the activation tensor that includes channels, height, width, and the precision (e.g., byte, float), is added to $\mu_i$, so $\mu_{i+1} \leftarrow \mu_i + \prod(u_i.$shape$)$.
Then we use $\mu_{i+1}$ as $\mu_{peak}$ to update $\mu_{peak,i+1}$ (peak memory footprint for $s_{i+1}$).
Since some predecessors of $u_i$ will not be used anymore after allocating $u_i$, we update the \emph{outdegrees} of the node by decrementing them.
Having updated the outdegree, we will be left with a \emph{zero-outdegree set} that denotes the nodes that are ready for deallocation.
We \emph{deallocate} the nodes in the set and update $\mu_{i+1}$ accordingly.

To demonstrate scheduling of a node $u_i$, Figure~\ref{fig:dp_memory} simulates scheduling a node $u_8$ = \circlewhite{\textsf{\size{8}{H}}} in $i = 8$.
In the figure, (1) \circlewhite{\textsf{\size{8}{H}}} is appended to $s_8$ and allocated to memory as it is scheduled, and then the scheduler records maximum of the $\mu_{peak,8}$ and the sum of all activations in the memory at this point as $\mu_{peak,9}$.
Then, it recalculates the outdegrees of the predecessor nodes of \circlewhite{\textsf{\size{8}{H}}}: \circlewhite{\textsf{\size{8}{D}}} and \circlewhite{\textsf{\size{8}{E}}}'s outdegree are decremented from one to zero.
(2) Then these nodes are deallocated and sum of the activation memory here is recorded as $\mu_9$.

\paragraph{Finding schedule with optimal peak memory footprint.}
After scheduling $u_i$, we save the new signature into the $\mathcal{M}_{i+1}$ for next search step $i+1$.
Since the goal of this work is to minimize the overall $\mu_{peak}$, we identify the corresponding optimal schedule $s^*_{i+1}$ for each $z_{i+1}$ by only saving $s_{i+1}$ with the minimum $\mu_{peak, i+1}$.
We integrate the aforementioned step of scheduling $u_i$ and updating $\mathcal{M}_{i+1}$ to complete the proposed \emph{dynamic programming-based scheduling algorithm}.
Algorithm~\ref{alg:dp_scheduling} summarizes the the algorithm.
As a first step, the algorithm starts by initializing the memoization table $\mathcal{M}_0$, then the algorithm iterates different search steps.
In each search step $i$, the algorithm performs the above illustrated memory allocation for all $u_i$ in $z_i$, and saving $s_{i+1}$, $\mu_{i+1}$, and $\mu_{peak, i+1}$.
After iterating all search steps to $n-1$, $s_*$ is saved in $\mathcal{M}_n$ with a unique entry, for $n$ being number of nodes in $\mathcal{G}$.
We provide the proof for the optimality of the peak memory footprint in the supplementary material.

\begin{algorithm}[t!]
   \caption{Dynamic Programming-based Scheduling}
   \label{alg:dp_scheduling}
\begin{algorithmic}[1]
	\STATE {\bfseries Input:} graph $\mathcal{G}$
	\STATE {\bfseries Output:} optimal schedule $s^*$
	\STATE // initialize memoization
	\STATE $s_0 \leftarrow []$, $\mu_0, \mu_{peak,0} \leftarrow 0$, $z_0 \leftarrow$ zero-indegree$(s_0, \mathcal{G})$
	\STATE $\mathcal{M}_0[z_0] \leftarrow (s_0, \mu_0, \mu_{peak,0})$
	\STATE // iterate search step
	\FOR{$i=0$ {\bfseries to} $n-1$}
		\STATE // iterate (schedule, current memory, peak memory)
		\FOR{$z_i, (s_i, \mu_i, \mu_{peak})$ {\bfseries in} $\mathcal{M}_i$}
			\FOR{$u_i$ {\bfseries in} $z_i$}
				\STATE $s_{i+1} \leftarrow s_i$.append$(u_i)$ // allocate
				\STATE $z_{i+1} \leftarrow$ zero-indegree$(s_{i+1}, \mathcal{G})$
				\STATE $\mu_{i+1}, \mu_{peak} \leftarrow \mu_i + \prod{(u_i.\text{shape})}$
				\STATE $\mu_{peak, i+1} \leftarrow \max(\mu_{peak, i}, \mu_{peak})$
				\FOR{$p_i$ {\bfseries in} $u_i.\text{preds}$}
					\IF{$p_i$ {\bfseries is in} zero-outdegree$(s_{i+1}, \mathcal{G})$}
					\STATE $\mu_{i+1} \leftarrow \mu_{i+1} - \prod{(p_i.\text{shape})}$ // deallocate
					\ENDIF
				\ENDFOR
				\STATE // memoize schedule with least peak memory
				\IF{$\mu_{peak, i+1} \leq \mathcal{M}_{i+1}[z_{i+1}].\mu_{peak, i+1}$}
				\STATE $\mathcal{M}_{i+1}[z_{i+1}] \leftarrow (s_{i+1}, \mu_{i+1}, \mu_{peak, i+1})$
				\ENDIF
			\ENDFOR
		\ENDFOR
	\ENDFOR
	\STATE $s^*, \mu_{peak}^* \leftarrow \mathcal{M}[\cdot]_{n}.s_n, \mathcal{M}[\cdot]_{n}.\mu_{peak,n}$ // solution
\end{algorithmic}
\end{algorithm}

\paragraph{Complexity of the algorithm.}
The complexity of the proposed dynamic programming-based scheduling is $\mathcal{O}(|V|\times2^{|V|})$, which is significantly faster than the exhaustive search of $\mathcal{S}_T$ with an upper bound complexity of $\mathcal{O}(|V|!)$.
Due to the space limitation, we present the derivation of the algorithm complexity in the supplementary material.

\subsection{Optimizing Scheduling Speed: Speeding up \\ the Dynamic Programming-based Scheduling}
\label{sec:optimized_implementation}

While the above scheduling algorithm improves complexity of the search, search space may still be intractable due to the immense irregularity.
Therefore, we devise \emph{divide-and-conquer} and \emph{adaptive soft budgeting} to accelerate the search by effectively shrinking and pruning the search space.

\begin{figure}[t!]
	\centering	
	\includegraphics[width=0.9\linewidth]{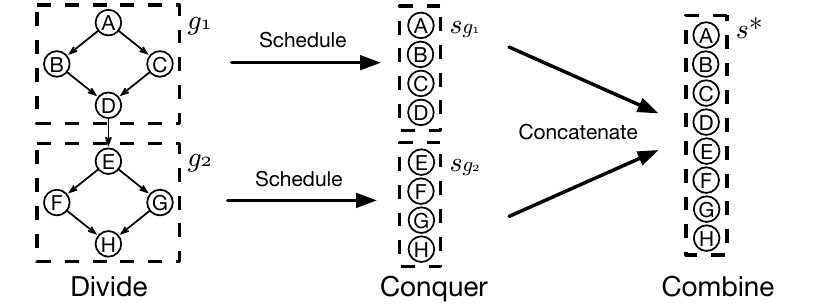}
	\vspace{-2ex}
	\caption{Illustration of \emph{divide-and-conquer}, which divides the graphs into multiple subgraphs (\emph{divide}), schedules each of them using the optimal scheduler (\emph{conquer}), then concatenates the sub-schedules to get the final schedule (\emph{combine}).}
	\vspace{-2ex}
	\label{fig:dnc}
\end{figure}

\paragraph{Divide-and-conquer.}
We can observe from Figure~\ref{fig:complex_networks} that the topology of irregularly wired neural networks are \emph{hourglass shaped} ( \raisebox{-1.5pt}{\rotatebox{90}{$\bowtie$}} ), because many NAS and Random Network Generators design \emph{cells} with single input and single output then \emph{stack them} to form an hourglass shape topology.
\cite{ilp_optimal:2000} shows that, during general purpose code scheduling, graphs can be partitioned (\emph{divide}) into multiple subgraphs and the corresponding solutions (\emph{conquer}) can be concatenated (\emph{combine}) to form an optimal solution for the overall problem.
While the complexity of the scheduling algorithm remains the same, this \emph{divide-and-conquer} approach can reduce the number of nodes in each subproblem, speeding up the overall scheduling time.
For instance, for a graph that can be partitioned into $N$ equal subgraphs, the scheduling time will decrease from $|V|\times2^{|V|}$ to $|V|\times2^{|V|/N}$ that we can speed up scheduling by multiple orders of magnitude compared to the naive approach, depending on the size of the graph and the number of partitions.

As such, Figure~\ref{fig:dnc} shows this insight can be extended to our problem setting, where we can first perform scheduling on each cell and merge those solutions together to form the final solution.
First, stage is partitioning the original graph $\mathcal{G}$ into multiple subgraphs $g$ (\emph{divide}).
Then, utilizing the independence among the subgraphs, each subgraph $g$ can be scheduled separately for their corresponding optimal schedule $s_g$ (\emph{conquer}).
Considering that the number of nodes in the subgraph $g$ is much smaller than the entire graph $\mathcal{G}$, the scheduling time will decrease significantly.
Finally, the schedules of the subgraphs are concatenated to give optimal schedule $s^*$ of the entire graph (\emph{combine}).

\begin{figure}[b!]
	\centering
	\subfigure[While both path $s_1$ and $s_2$ schedules lead to same $z'$, their $\mu$ and $\mu_{peak}$ varies and we can prune schedules that yield higher $\mu_{peak}$ than a given budget $\tau$. Numbers next to box or circle are $\mu$ and numbers next to edges are $\mu_{peak}$]{\includegraphics[width=0.9\linewidth]{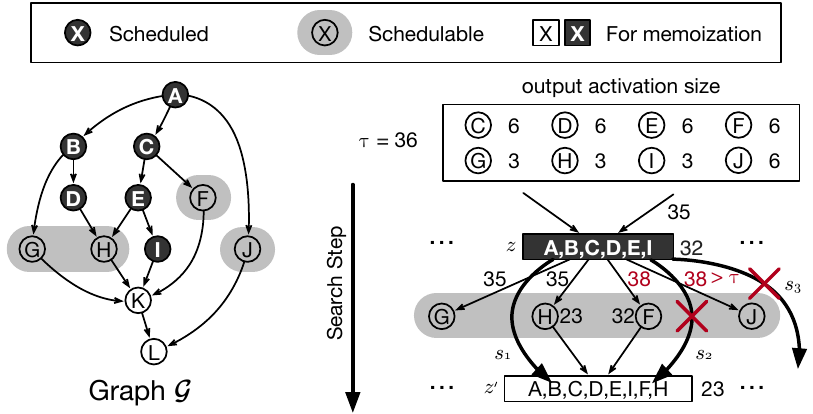}\label{fig:pruning_paths}}
	\hfill
	\subfigure[\emph{Adaptive soft budgeting} starts by setting a hard budget $\tau_{max}$ as the maximum value for the soft budget $\tau$. Then, conducts a binary search for $\tau$, higher than $\tau^*$ that it finds a solution yet not too high that scheduling completes quickly.]{\includegraphics[width=0.9\linewidth]{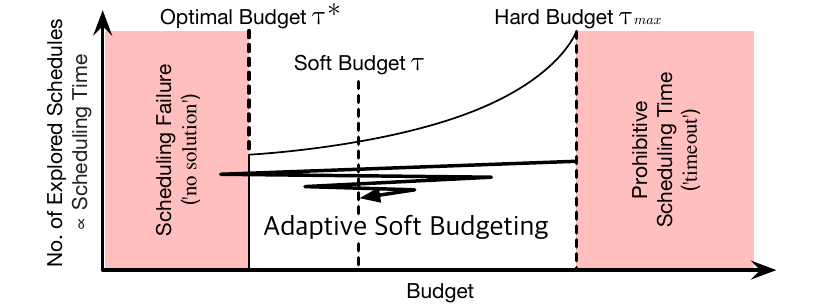}\label{fig:pruning_threshold}}
	\vspace{-2ex}
	\caption{Illustration of the \emph{adaptive soft budgeting}. (a) shows how schedules are pruned, and (b) illustrates how the \emph{soft budget} $\tau$ relates to the number of explored schedules.}
	\label{fig:pruning}
\end{figure}

\paragraph{Adaptive soft budgeting.}
While divide-and-conquer approach scales down the number of nodes, the algorithm may still not be fast enough due to the exponential complexity of the algorithm.
Therefore, we explore avoiding suboptimal solutions during the early stage of scheduling without affecting the optimality of the original algorithm.
Since our goal is to find a \emph{single} solution that can run within a given \emph{memory budget} $\tau^* = \mu^*$ while all other solutions can be discarded, setting some budget $\tau$ that is greater or equal to $\mu^*$ and pruning suboptimal schedules with which their $\mu_{peak}$ exceeds $\tau$ can focus the search to a smaller search space $\mathcal{S}_T' \subset \mathcal{S}_T$ while still achieving the optimal schedule $s^*$.
On top of this, we develop a \emph{meta-search} for $\tau$.
This is inspired from engineers buying a larger memory (increase $\tau$) if a program fails due to stack overflow (= \textsf{'no solution'} due to an overly aggressive pruning) and selling out excess memory (decrease $\tau$) if the current budget is prohibitive (= \textsf{'timeout'} due to lack of pruning).
\name takes advantage of this insight to develop an \emph{adaptive soft budgeting} scheme while scheduling to cut down the overall number of explored schedules.
Figure~\ref{fig:pruning} illustrates the overall idea by first showing how some schedules are pruned with regard to a given budget $\tau$ in Figure~\ref{fig:pruning_paths} then implication of different $\tau$ on scheduling time in  Figure~\ref{fig:pruning_threshold}.

\begin{algorithm}[b!]
	\caption{Adaptive Soft Budgeting}
	\label{alg:pruning}
\begin{algorithmic}[1]
   \STATE {\bfseries Input:} graph $\mathcal{G}$
   \STATE {\bfseries Output:} optimal schedule $s^*$
   \STATE $\tau_{max} \leftarrow \mu($Kahn'sAlgorithm$(\mathcal{G}), \mathcal{G})$ // hard budget
   \STATE $\tau_{old}, \tau_{new} \leftarrow \tau_{max}$
   \STATE $flag \leftarrow$ \textsf{'no solution'}
   \REPEAT
   \STATE // binary search for $\tau$: decrease $\tau$ if 'timeout' 
   \STATE // and increase $\tau$ if \textsf{'no solution'}
   \IF{$flag$ {\bfseries is} \textsf{'timeout'}}
   \STATE // simultaneous
   \STATE $\tau_{old} \leftarrow \tau_{new}$, $\tau_{new} \leftarrow \tau_{new} / 2$
   \ELSIF{$flag$ {\bfseries is} \textsf{'no solution'}}
   \STATE // simultaneous
   \STATE $\tau_{old} \leftarrow \tau_{new}$, $\tau_{new} \leftarrow (\tau_{new} + \tau_{old}) / 2$
   \ENDIF
   \IF{$flag$  {\bfseries is} \textsf{'solution'}}
   \STATE $s^* \leftarrow schedule$ // optimal schedule
   \ENDIF
   \UNTIL{$flag$ {\bfseries is} \textsf{'solution'}}
\end{algorithmic}
\end{algorithm}

Figure~\ref{fig:pruning_paths} depicts a certain point while scheduling $\mathcal{G}$, where nodes \circlewhite{\textsf{\size{8}{G}}}, \circlewhite{\textsf{\size{8}{H}}}, \circlewhite{\textsf{\size{8}{F}}}, and \circlewhite{\textsf{\size{8}{J}}} can be scheduled.
In particular, the figure compares two possible solutions $s_1$ and $s_2$ which schedules \circlewhite{\textsf{\size{8}{H}}}$\rightarrow$\circlewhite{\textsf{\size{8}{F}}} and \circlewhite{\textsf{\size{8}{F}}}$\rightarrow$\circlewhite{\textsf{\size{8}{H}}}, respectively given $\tau=36$.
While $s_1$ and $s_2$ both starts from $z$ with $\mu=32$, scheduling \circlewhite{\textsf{\size{8}{H}}} leads to $\mu_{peak} = 32 + 3$ (H) $= 35$, whereas scheduling \circlewhite{\textsf{\size{8}{F}}} or \circlewhite{\textsf{\size{8}{J}}} leads to $\mu_{peak} = 32 + 6$ (F or J) $= 38$. 
Therefore, since we assume $\tau = 36$, $s_2$ and $s_3$ will fail because $\mu_{peak} = 38$ for $s_2$  and $s_3$ exceeds $36$.
So, as long as we set the budget $\tau$ higher than $\mu^*$, the scheduler still finds a single optimal solution while avoiding many suboptimal paths.
On the other hand, too small a $\tau < \mu^*$ leads to no solution because the optimal path would be pruned away.

Having established the possibility of pruning, our question boils down to discovering $\tau$ that is greater or equal to $\mu^*$ which we call an \emph{optimal budget} $\tau^*$, yet close enough to shrink the search space effectively.
Figure~\ref{fig:pruning_threshold} and Algorithm~\ref{alg:pruning} summarizes the proposed \emph{adaptive soft budgeting}.
Since we start with no information about the approximate range for $\tau$, we resort to a commonly used topological ordering algorithm called Kahn's algorithm~\cite{kahn:1962} ($\mathcal{O}(|V|+|E|)$) to adaptively gain idea of the range for $\tau$.
We use the peak memory footprint from this sequence and use it as our \emph{hard budget} $\tau_{max}$, and in contrast we call adaptively changing $\tau$ as a \emph{soft budget}.
Since $\tau_{max} \geq \mu^*$, we know that any $\tau \ge \tau_{max}$ do not need to be explored.
Having this upper bound for the search, \emph{adaptive soft budgeting} implements a binary search to first run the scheduling algorithm with $\tau$ and $T$ as input, where $T$ is an hyperparameter that limits the scheduling time per search step.
The binary search increases $\tau$ ($\tau_{new} \leftarrow (\tau_{new} + \tau_{old}) / 2$) if it finds \textsf{'no solution'} and decreases $\tau$ ($\tau_{new} \leftarrow \tau_{new} / 2$) if a search step returns \textsf{'timeout'} (search step duration exceeds $T$).
The binary search stops as soon as it finds a schedule (\textsf{'solution'}), and this method using binary search is guaranteed to work due to the monotonically increasing number of explored schedules with $\tau$.

\begin{figure}[t!]
	\centering	
	\includegraphics[width=0.9\linewidth]{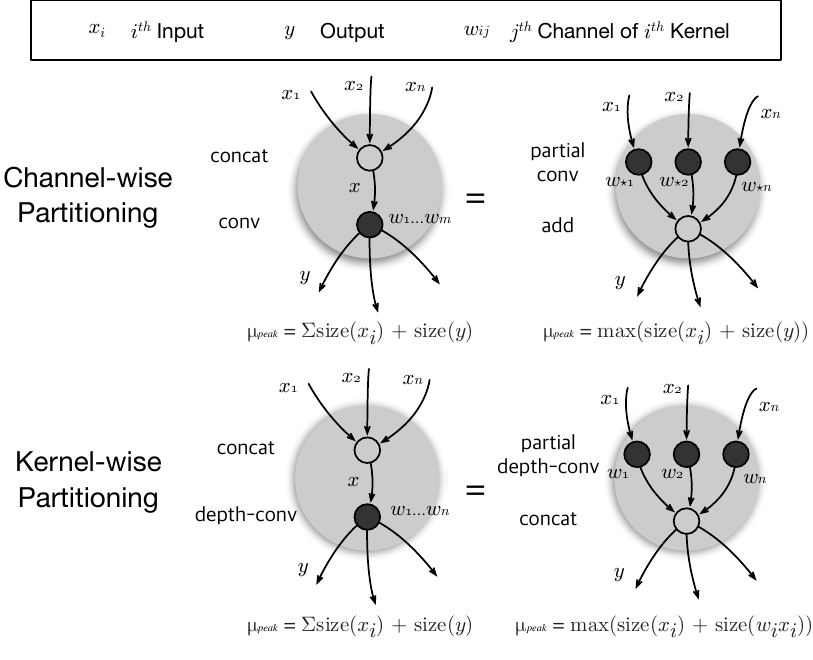}
	\vspace{-2ex}
	\caption{Illustration of the graph rewriting patterns: \emph{channel-wise partitioning} and \emph{kernel-wise partitioning} can reduce the memory cost of convolution and depthwise convolution respectively.}
	\vspace{-2ex}
	\label{fig:rewriting}
\end{figure}

\subsection{Identity Graph Rewriting: Improving the Search Space for Better Peak Memory Footprint}
\label{sec:graph_rewriting}
%
Reorganizing the computational graph of the irregularly wired neural networks may lead to \emph{significant reduction in the peak memory footprint} $\mu_{peak}$ during computation.
%
For example, it is notable that large stream of NAS-based works~\cite{darts:2018,swiftnet:2019} rely on extensive use of \emph{concatenation} as a natural approach to merge information from multiple branches of the input activations and expand the search space of the neural architectures.
However, concatenation with many incoming edges may prolong the liveness of the input activation and increase the memory pressure, which is unfavorable especially for resource constrained scenarios.
%
%
To address this issue, we propose \emph{identity graph rewriting} to effectively reduce $\mu_{peak}$ around the concatenation while keeping the arithmetic outputs identical.
To this end, we present two main examples of the graph patterns in irregularly wired neural networks that benefits from our technique:

\paragraph{Channel-wise partitioning (convolution).}
One typical pattern in irregularly wired neural networks is \emph{concatenation} (\emph{concat}: $[\cdot]$) that takes multiple branches of the input prior to a \emph{convolution} (\emph{conv}: $*$).
%
While executing such pattern, peak memory footprint $\mu_{peak}$ occurs when the output $y \in \mathbb{R}^n$ is being computed while concatenated branches of input $x \in \mathbb{R}^n$ are also mandated to reside in the memory.
Our objective is to achieve the same arithmetic results and logical effect as \emph{concat} yet sidestep the corresponding seemingly excessive memory cost.
To this end, we \emph{channel-wise partition} the \emph{conv} that follows the \emph{concat} so that the \emph{partitioned conv} can be computed as soon as the input $x_i$ becomes available.
Equation~\ref{eq:partial_conv_begin}-\ref{eq:partial_conv_end} detail the mathematical derivation of this substitution.
Specifically, as shown in Equation~\ref{eq:partial_conv_begin}, each kernel iterates and sums up the result of convolving channels in conv.
However, using the \emph{distributive property} of $\sum_i$ and $*$, these transform to summation of \emph{channel-wise partitioned convolution}, which we call \emph{partial conv}.
This \emph{partial conv} removes \emph{concat} from the graph leading to lower \emph{memory cost}.
As illustrated in Figure~\ref{fig:rewriting}, the memory cost of same computation reduces from $\sum x_i + y$ to $max(w_{\star i}*x_i) + y$, which becomes more effective when there are more incoming edges to concat.

\vspace{-2ex} 
\begin{align}
\label{eq:partial_conv_begin}
y &= \Big[\sum_i w_{1i} * x_i, ..., \sum_i w_{mi} * x_i\Big] \;\; \text{(concat+conv)} \\
  &= \sum_i \Big[w_{1i} * x_i, ..., w_{mi} * x_i\Big] \\
  &= \sum_i \Big[w_{1i}, ..., w_{mi}\Big] * x_i \\
  \label{eq:partial_conv_end}
  &= \sum_i \Big[w_{\star i} * x_i\Big] \qquad\qquad\qquad \text{(partial conv+add)}
\end{align}

\paragraph{Kernel-wise partitioning (depthwise convolution).}
\emph{Depthwise convolution} (\emph{depthconv})~\cite{depthwise:2014,mobilenet:2017} has been shown to be effective in reducing computation yet achieve competitive performance, hence its wide use in networks that target extreme efficiency as its primary goal.
For \emph{concatenation} (\emph{concat}) followed by a \emph{depthwise convolution} (\emph{depthconv}), similar to above \emph{concat+conv} case, peak memory footprint $\mu_{peak}$ occurs when the concatenated $x$ is inside the memory and the result $y$ additionally gets saved to the memory before $x$ is deallocated.
This time, we leverage the \emph{independence} among different kernels to \emph{kernel-wise partition} the \emph{depthconv} that follows the \emph{concat} so that each input $x_i$ is computed to smaller feature maps without residing in the memory too long.
As such, Equation~\ref{eq:partial_depthconv_begin}-\ref{eq:partial_depthconv_end} derives this substitution.
Equation~\ref{eq:partial_depthconv_begin} shows that every component in the $y$ is \emph{independent} (different subscript index) and is viable for partitioning.
In other words, this rewriting simply exposes the \emph{commutative property} between \emph{depthconv} and \emph{concat} plus \emph{kernel-wise partitioning} to reduce $\mu_{peak}$ significantly.

\vspace{-2ex} 
\begin{align}
\label{eq:partial_depthconv_begin}
y &= \Big[w_1 * x_1, ..., w_n * x_n\Big] \qquad \text{(concat+depthconv)} \\
  \label{eq:partial_depthconv_end}
  &= \Big[[w_1 * x_1], ..., [w_n * x_n]\Big] \;\;\, \text{(partial depthconv+concat)}
\end{align}

\paragraph{Implementation.}
Following the general practice of using pattern matching algorithms in compilers~\cite{llvm:2004,glow:2018,graph_substitution:2019}, we implement \emph{identity graph rewriting} using pattern matching to identify regions of the graph which can be substituted to an operation with \emph{lower computational cost}.
Likewise, we make use of this technique to identify regions that leads to \emph{lower memory cost}.

\section{Evaluation}
\label{sec:evaluation}

We evaluate \name with four representative \emph{irregularly wired neural networks} graphs.
We first compare the peak memory footprint of \name against TensorFlow Lite~\cite{tflite} while using the same linear memory allocation scheme\footnote{TensorFlow Lite implements a linear memory allocator named simple memory arena: \url{https://github.com/tensorflow/tensorflow/blob/master/tensorflow/lite/simple_memory_arena.cc}} for both.
Furthermore, we also experiment the impact of such peak memory footprint reduction on off-chip memory communication.
%
%
We also conduct an in-depth analysis of the gains from the proposed \emph{dynamic programming-based scheduler} and \emph{graph rewriting} using SwiftNet Cell A~\cite{swiftnet:2019}.
Lastly, we study the impact of \emph{adaptive soft budgeting} on the scheduling time.

\subsection{Methodology}
\label{sec:eval_method}
\begin{table}[t!]
\vspace{-2ex} 
\caption{Specification of the networks used for evaluation.}
\label{tab:networks}
\vspace{-2ex} 
\begin{center}
\begin{scriptsize}
\begin{sc}
\begin{tabular}{lccccc}
\toprule
Network      & Type   & Dataset & \# MAC & \# Weight & Top-1    \\
             &        &         &        &           & Accuracy \\
\midrule
DARTS        & NAS    & ImageNet & 574.0M &   4.7M  & 73.3\% \\
SwiftNet     & NAS    & HPD      &  57.4M & 249.7K  & 95.1\% \\
RandWire     & Rand   & CIFAR10  & 111.0M &   1.2M  & 93.6\% \\
RandWire     & Rand   & CIFAR100 & 160.0M &   4.7M  & 74.5\% \\
\bottomrule
\end{tabular}
\end{sc}
\end{scriptsize}
\end{center}
\vspace{-4ex} 
\end{table}

\paragraph{Benchmarks and datasets.}
Table~\ref{tab:networks} lists the details of the networks--representative of the irregularly wired neural networks from Neural Architecture Search (\textsc{NAS}) and Random Network Generators (\textsc{Rand})--used for evaluation: DARTS~\cite{darts:2018} for ImageNet, SwiftNet~\cite{swiftnet:2019} for a dataset comprised of human presence or absence (HPD), and RandWire~\cite{randomnet:2019} for CIFAR10 and CIFAR100.
DARTS~\cite{darts:2018} is a gradient-based NAS algorithm.
In particular we focus on the learned normal cell for image classification on ImageNet: only the first cell because it has the highest peak memory footprint and the reset of the network is just repeated stacking of the same cell following the practice in NASNet~\cite{nasnet:2018}.
SwiftNet~\cite{swiftnet:2019} is network from NAS by targeting human detection dataset.
RandWire~\cite{randomnet:2019} are from Random Network Generators for image classification on CIFAR10 and CIFAR100.
The table also lists their dataset, multiply-accumulate count (\# \textsc{MAC}), number of parameters (\# \textsc{Weight}), and top-1 accuracy on their respective dataset.
%

\subsection{Experimental Results}
\label{sec:eval_results}

\paragraph{Comparison with TensorFlow Lite.}
Figure~\ref{fig:eval_overall} evaluates \name over TensorFlow Lite on different cells of the aforementioned networks in terms of reduction in memory footprint.
The figures illustrate that \name's dynamic programming-based scheduler reduces the memory footprint by a factor of \DPPeakImprovement without any changes to the graph.
In addition, the proposed graph rewriting technique yields an average of \PConvPeakImprovement (extra \PConvPeakExtraImprovement) reduction in terms of peak memory footprint.
The results suggest that \name yields significant reduction in terms of the peak memory footprint for irregularly wired neural networks.

\begin{figure}[t!]
	\centering
	\vspace{-1ex} 
	\includegraphics[width=0.9\linewidth]{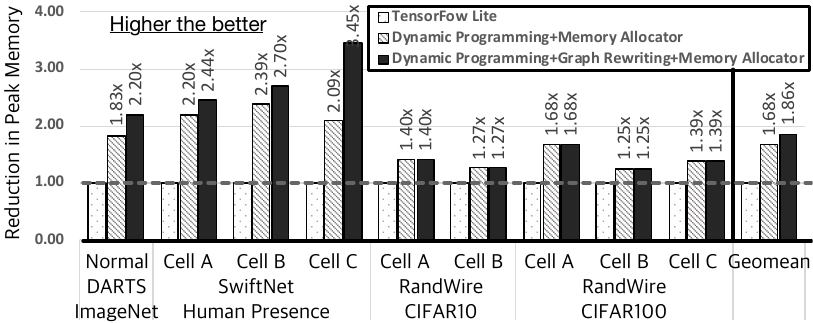}\label{fig:eval_memory}
	\vspace{-2ex}
	\caption{Reduction in peak memory footprint of \name against TensorFlow Lite (no memory hierarchy).}
	\vspace{-2ex}
	\label{fig:eval_overall}
\end{figure}

\begin{figure}[t!]
	\centering
	\includegraphics[width=0.9\linewidth]{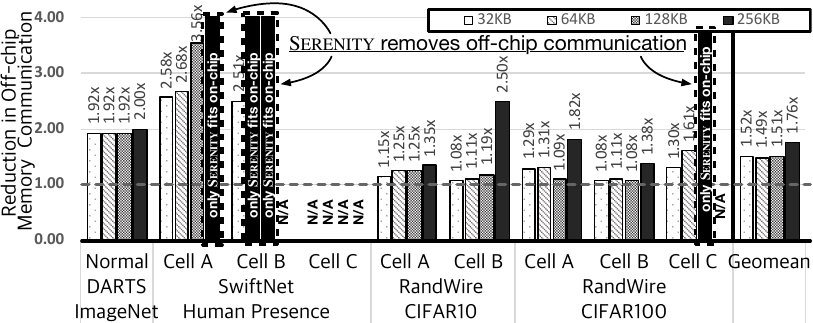}\label{fig:eval_memory}
	\vspace{-2ex}
	\caption{Reduction in off-chip memory communication of \name against TensorFlow Lite (with memory hierarchy).}
	\vspace{-3ex}
	\label{fig:eval_offchip}
\end{figure}

\paragraph{Improvement in off-chip memory communication.}
We also show how \name affects the off-chip memory communication, which largely affects both power and inference speed~\cite{eyeriss:2016,tetris:2017,bitfusion:2018}.
To this end, Figure~\ref{fig:eval_offchip} sweeps different on-chip memory configurations to measure the reduction in off-chip communication on systems with multi-level memory hierarchy.
Since we know the entire schedule \emph{a priori}, we use Belady's optimal algorithm~\cite{belady:1966}, also referred to as the clairvoyant algorithm for measuring the off-chip memory communication, to distill the effects of the proposed scheduling.
The results show that \name can reduce the off-chip memory communication by \OffChipImprovement for a device with \textsf{256KB} on-chip memory.
In particular, while there were few cases where peak memory footprint was already small enough to fit on-chip (\textsf{N/A} in figure), there were some cases where \name eradicated the off-chip communication by successfully containing the activations in the on-chip memory while TensorFlow Lite failed to do so (marked in figure).
This suggests that \name's effort of reducing memory footprint is also effective in reducing the off-chip memory communication in systems with memory hierarchy, hence the power consumption and inference speed.

\paragraph{Improvement from dynamic programming-based scheduler and identity graph rewriting.}
To demonstrate where the improvement comes from, Figure~\ref{fig:eval_footprint_reason} plots the memory footprint while running Swiftnet Cell A.
Figure~\ref{fig:eval_ma_footprint} shows the memory footprint of \name with the memory allocation.
The figure shows that \name's dynamic programming-based scheduler brings significant improvement to the peak memory footprint (\textsf{551.0KB$\rightarrow$250.9KB}), and the graph rewriting further improves this by \textsf{25.1KB} (\textsf{250.9KB$\rightarrow$225.8KB}) by utilizing patterns that alleviate regions with large memory footprint.
In order to focus on the effect of the scheduler and graph rewriting, Figure~\ref{fig:eval_footprint} presents the memory footprint of \name without the memory allocation: the sum of the activations while running the network.
The figure shows that the proposed scheduler finds a schedule with the optimal (minimum) peak memory footprint without changes to the graph.
Then, it shows that the proposed graph rewriting can further reduce the peak memory footprint by \textsf{12.5KB} (\textsf{200.7KB$\rightarrow$188.2KB}).
The results suggest that the significant portion of the improvement comes from the proposed dynamic programming-based scheduler and the graph rewriting.

\begin{figure}[t!]
	\centering
	\vspace{-1ex} 
	\subfigure[Memory footprint with the memory allocator (peak memory footprint of TensorFlow Lite = \textsf{551.0KB}).]{\includegraphics[width=0.9\linewidth]{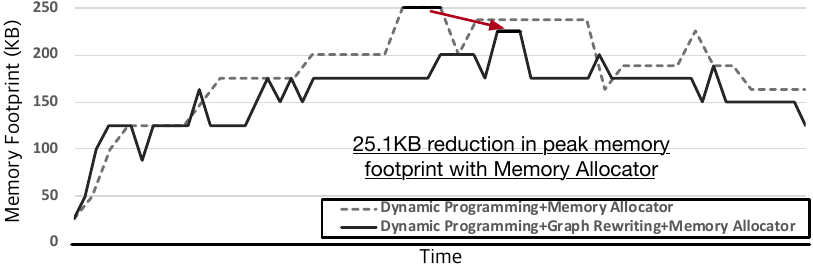}\label{fig:eval_ma_footprint}}
	\\[-0.5ex]
	\subfigure[Memory footprint without the memory allocator.]{\includegraphics[width=0.9\linewidth]{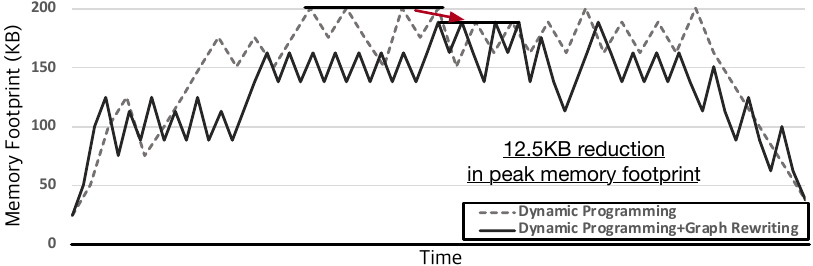}\label{fig:eval_footprint}}
	\vspace{-2ex}
	\caption{Memory footprint while running SwiftNet Cell A with and without the memory allocator (\emph{red arrow} denotes reduction).}
	\vspace{-3ex}
	\label{fig:eval_footprint_reason}
\end{figure}

\paragraph{Scheduling time of \name.}
Figure~\ref{fig:eval_time} summarizes the (static) scheduling time taken for \name to schedule the networks.
Results show that the average scheduling time is \DPSchedulingLatency without the graph rewriting and \PConvSchedulingLatency with graph rewriting, which the difference comes from the increase in the number of nodes from graph rewriting.
The results show that all the above gains of \name come at the cost of less than one minute average extra compilation time. 
While the dynamic programming-based scheduling suffers from an exponential time complexity, \name manages to make the scheduling tractable through the proposed divide-and-conquer and adaptive soft budgeting.

\paragraph{Speed up from divide-and-conquer and adaptive soft budgeting.}
Table~\ref{tab:time} summarizes the scheduling time of SwiftNet~\cite{swiftnet:2019} for different algorithms to demonstrate the speed up from divide-and-conquer and adaptive soft budgeting techniques.
As such, the table lists different combination of algorithms, number of nodes, and the corresponding scheduling time.
Straightforward implementation of the aforementioned \circleblack{1} \emph{dynamic programming-based scheduling} leads to an immeasurably large scheduling time regardless of the \emph{graph rewriting}.
However, additional application of the \circleblack{2} \emph{divide-and-conquer} (\circleblack{1}+\circleblack{2}) leads to a measurable scheduling time: \textsf{56.53 secs} and \textsf{7.29 hours} to schedule without and with the graph rewriting, respectively.
Furthermore, we observe that further applying \circleblack{3} \emph{adaptive soft budgeting} (\circleblack{1}+\circleblack{2}+\circleblack{3}) significantly reduces the scheduling time \textsf{37.9 secs} and \textsf{111.9 secs} to schedule without and with the graph rewriting, respectively.
Above results indicate that applying the proposed algorithms leads to a scheduling time of practical utility.

\begin{figure}[t!]
	\centering
	\vspace{-1ex} 
	\includegraphics[width=0.95\linewidth]{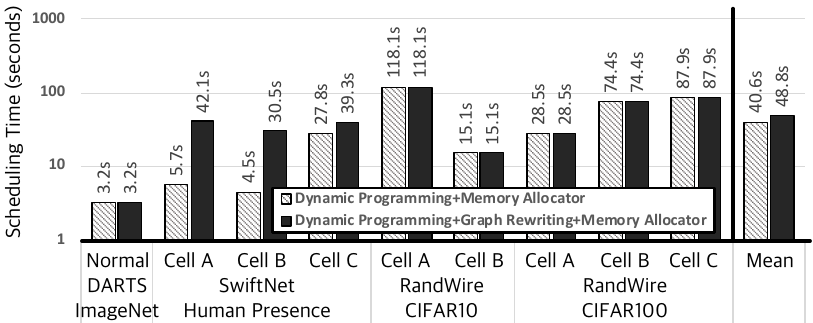}
	\vspace{-2ex}
	\caption{Scheduling time evaluation for \name.}
	\vspace{-3ex}
	\label{fig:eval_time}
\end{figure}

\begin{table}
\caption{Comparison of the scheduling time for different algorithms to schedule SwiftNet. \circleblack{1}, \circleblack{2}, and \circleblack{3} represent \emph{dynamic programming}, \emph{divide-and-conquer}, and \emph{adaptive soft budgeting} respectively. \textsf{N/A} denotes infeasible within practical time.}
\label{tab:time}
\begin{center}
\begin{scriptsize}
\begin{sc}
\begin{tabular}{lllrr}
\toprule
Graph     & Algorithm  & \# Nodes and & \multicolumn{1}{l}{Scheduling} \\
Rewriting &            & Partitions &   \multicolumn{1}{l}{Time}             \\
\midrule
\multicolumn{1}{c}{\xmark} &\circleblack{1} & 62 =\{62\}       & \textsf{N/A}    \\
\multicolumn{1}{c}{\xmark} &\circleblack{1}+\circleblack{2}        & 62=\{21,19,22\} & \textsf{56.5 \textup{secs}}                    \\
\multicolumn{1}{c}{\xmark} &\circleblack{1}+\circleblack{2}+\circleblack{3}     & 62=\{21,19,22\} & \textbf{\textsf{37.9 \textup{secs}}}           \\
\midrule
\multicolumn{1}{c}{\cmark} &\circleblack{1}        & 92=\{92\}       & \textsf{N/A}    \\
\multicolumn{1}{c}{\cmark} &\circleblack{1}+\circleblack{2}     & 92=\{33,28,29\} & \textsf{7.2 \textup{hours}}      \\
\multicolumn{1}{c}{\cmark} &\circleblack{1}+\circleblack{2}+\circleblack{3}  & 92=\{33,28,29\} & \textbf{\textsf{111.9 \textup{secs}}}         \\
\bottomrule
\end{tabular}
\end{sc}
\end{scriptsize}
\end{center}
\vspace{-2ex} 
\end{table}





\section{Related Works}
\label{sec:related}

The prevalence of neural networks has led to the development of several compilation frameworks for deep learning~\cite{tensorflow:2016,pytorch:2017,glow:2018,ngraph:2018}.
However, even industry grade tools, mostly focus on tiling and fine-grained scheduling of micro-operations on the conventional hardware~\cite{tensorrt,tflite} or accelerators~\cite{eyeriss:2016,dadiannao:2014,eie:2016,stripes:2016,tpu:2017,tetris:2017,scnn:2017,bitfusion:2018,brainwave:2018}.
However, these framework are mostly designed for the common regular patterns that have dominated deep learning from almost its conception.
As such, these tools inherently had no incentive to deal with the form of irregularities that the emerging NAS~\cite{nas:2016,adanet:2017,nasnet:2018,darts:2018,proxylessnas:2018,amoebanet:2018,swiftnet:2019} and Random Networks~\cite{randomnet:2019,wirings:2019} bring about.
This paper, in contrast, focuses on this emergent class that breaks the regularity convention and aims to enable their execution on memory constrained edge devices.

\paragraph{Scheduling and tiling for neural networks.}
While prior works on scheduling~\cite{vliw_scheduling:2013,dp_optimal:2001,ilp_optimal:2000} focus on classical computing workloads, there have been limited study about the implications of scheduling in the neural networks domain.
There is also a significant body of work on scheduling operations on hardware accelerators~\cite{dla:2018} that also considers tiling~\cite{tvm:2018,tc:2018,pbqp:2019,chameleon:2020}.
However, graph scheduling  for irregularly wired neural network, specially with memory constraints, is an  emerging problem, which is the focus of this paper.

\paragraph{Graph rewriting for neural networks.}
It has been a common practice to rewrite parts of the graph using rule-based~\cite{tensorflow:2016,pytorch:2017,glow:2018,ngraph:2018,tensorrt} or systematic approaches to expose parallelism and make models more target-aware~\cite{hidden_dimension:2018,graph_substitution:2019,hw_dependent_rewriting:2007}.
While these approaches may alleviate the complexity of the graph and reduce the peak memory footprint as a side effect, these frameworks do not explore and are not concerned with scheduling.
Our work exclusively explores graph rewriting in the context of improving the peak memory footprint.

\paragraph{Optimizing neural networks.}
There are different optimization techniques that aim to simplify the neural network indifferent dimensions.
Sparsification/compression~\cite{brain_damage:1990,learning_connections:2015,prune:2017,structured_pruning:2017}, 
quantization~\cite{deep_compression:2015,bnn:2016,dorefa:2016,apprentice:2017,lsq:2019}, 
activation compression~\cite{gist:2018},
and kernel modifications reduce the complexity of the individual operations or remove certain computations.
However, our focus, the problem of memory-aware graph scheduling still remains orthogonal to these inspiring efforts.
%



\section{Conclusion}
\label{sec:conclusion}
As the new forms of connectivity emerges in neural networks, there is a need for system support to enable their effective use, specially for intelligence at the edge.
This paper took an initial step toward orchestrating such network under stringent physical memory capacity constraints.
We devised signatures to enable dynamic programming and adaptive soft budgeting to make the optimization tractable.
Even more, an identity graph writing was developed to further the potential for gains.
The encouraging results for a set of emergent networks suggest that there is significant potential for compiler techniques that enables new forms of intelligent workloads.
%
%
\section*{Acknowledgement}
We thank the anonymous reviewers for their insightful comments.
We also thank Harris Teague and Jangho Kim for the fruitful discussions and feedbacks on the manuscript, and Parham Noorzad for his help with the mathematical formulations to calculate the complexity of the algorithms.

{
\footnotesize
\bibliographystyle{formatting/mlsys2020}
\bibliography{ref_camera}
}

\vfill
\clearpage
\appendix

\section{Comparison between Irregularly Wired Neural Networks and Conventional Regular Topology Neural Networks}
\label{sec:vwwc}

\begin{figure}[H]
	\centering
	\vspace{-2ex} 
	\subfigure[ImageNet accuracy vs number of multiply-and-accumulate.]{\includegraphics[width=0.95\linewidth]{figs/networks.pdf}}
	\subfigure[ImageNet accuracy vs number of parameters.]{\includegraphics[width=0.95\linewidth]{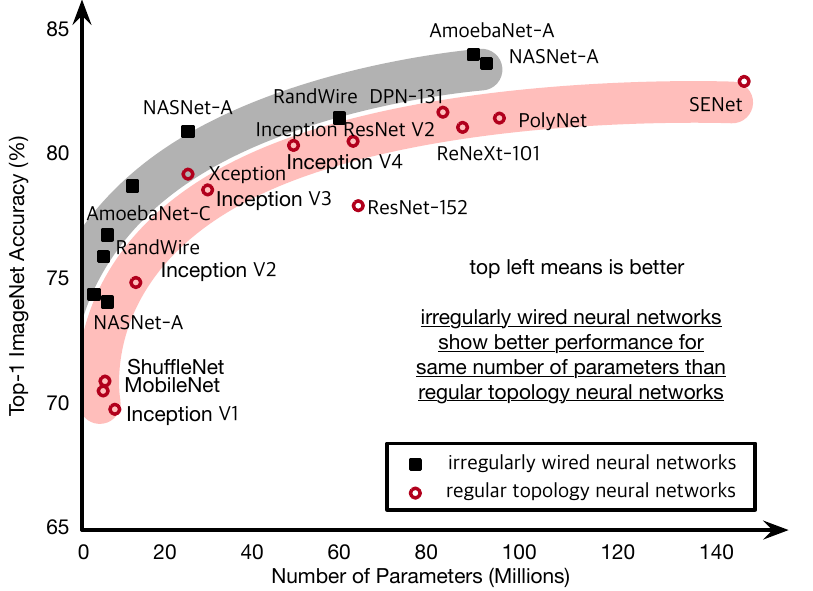}}
	\vspace{-2ex}
	\caption{ImageNet accuracy vs number of multiply-and-accumulate or parameters, where irregularly wired neural networks show higher performance for same amount of compute or number of parameters than regular topology neural networks.}
	\vspace{-2ex}
\end{figure}
\section{Comparison with TensorFlow Lite}
In addition to the relative reductions provided in Figure~\ref{fig:eval_overall}, Figure~\ref{fig:comparison} provides the raw numbers of the peak memory footprint for the benchmark irregularly wired neural networks.

\begin{figure}[h]
	\centering
	\vspace{-1ex} 
	\includegraphics[width=0.9\linewidth]{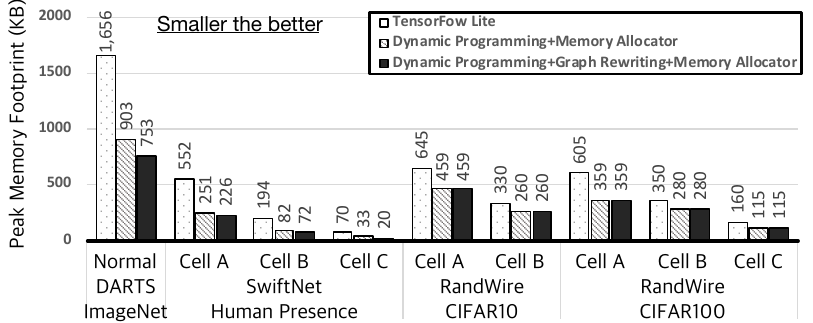}
	\vspace{-2ex}
	\caption{Peak memory footprint of running irregularly wired neural networks on \name and TensorFlow Lite.}
	\vspace{-2ex}
	\label{fig:comparison}
\end{figure}

\section{Proof for Optimal Peak Memory Footprint from the Dynamic Programming-based Scheduling}
Here we prove the optimality of the above dynamic programming-based scheduling algorithm.

\begin{theorem}
In order to find a schedule $s^*$ with an optimal peak memory consumption $\mu^*$, it is sufficient to keep just one schedule-peak memory pair ($s_i$, $z_i$) in $S_{Ti}$ for each zero-indegree set $z_i$, and to append subsequent nodes on top of $s_i$ to get $s_{i+1}$ in each search step. 
\end{theorem}

\begin{proof}
If $i = 0$, the optimal $s_0$ is an empty sequence and $\mu_0$ must be 0.
On the other hand, if $i \geq 1$, assume that (\emph{suboptimal}) $v_i$ constitutes $s^*$, substituting $u_i^* \in z_i$ and achieves $\mu^*$.
In such case, let $v_i$ be replaced with (\emph{optimal}) $u_i^*$, which will result in $\mu_{peak} \leftarrow \min(\mu_i + \prod v_i.\text{shape}, \mu_i + \prod u_i^*.\text{shape})$, and $\mu_{i+1}$ is calculated by deducting $\prod p_i.\text{shape}, \forall p_i \in (u_i.\text{preds} \cap \text{zero-outdegree}(s_{i+1}, \mathcal{G}))$.
By recursively applying $u_k$ for rest of the search steps $k$, the algorithm should find an alternative sequence $s^*$$'$ with $\mu^*$$'$$ \leq \mu^*$ due to the $\min$ operator above, contradicting the original assumption on the optimality of $s^*$.
Therefore, our algorithm finds a schedule with an optimal peak memory consumption.
\end{proof}
\section{Complexity Analysis of the Dynamic Programming-based Scheduling  and Proof}
\label{sec:complexity}
We compare the complexity of exhaustively exploring $\mathcal{S}_T$ and our dynamic programming-based scheduling.
While the algorithm both lists candidate schedules and calculates their peak memory footprint, we consider the peak memory footprint calculation as one operation while deriving the complexity.
In order to visualize the analysis, we invent $\mathcal{G}$ in Figure~\ref{fig:complexity_toy} to demonstrate the upper bound complexity of each algorithm.
It has a single entry node and a single exit node \circlewhite{\textsf{\size{8}{A}}} and \circlewhite{\textsf{\size{8}{Z}}}, respectively, and all other nodes constitute independent branches between the entry and the exit node.

\begin{figure}[H]
	\centering	
	\includegraphics[width=0.95\linewidth]{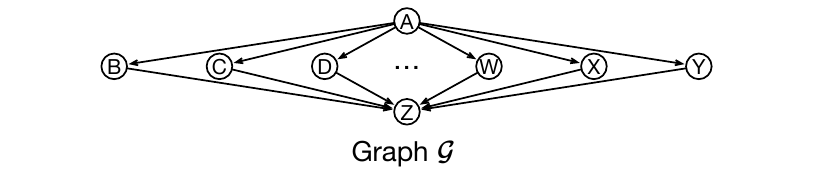}
	\vspace{-2ex}
	\caption{Topology of $\mathcal{G}$ to demonstrate the upper bound complexity of each algorithm.}
	\vspace{-2ex}
	\label{fig:complexity_toy}
\end{figure}

First, we demonstrate the complexity of the recursive topological sorting that exhaustively explores $\mathcal{S}_T$.
Since there is a single entry node and a single exit node, there will be $|V-2|$ remaining nodes and these nodes can be scheduled independently of one another, thereby the number of candidate schedules become $\langle|V-2|!\rangle$ and the overall complexity becomes $\mathcal{O}(|V|!)$, where $|V|$ denotes the number of nodes.
On the other hand, for the dynamic programming we calculate the number of candidates by utilizing the number of schedules that gets memoized.
Our memoization takes advantage of the zero-indegree sets $z$ for each search step.
%

For the first and the last search steps, we assume that we have a single entry node and a single exit node.
On the other hand, since the number of nodes scheduled in search step $i$ would be $i-1$, the maximum number of entries for memoization is $\binom{|V|-2}{i-1}$.
On top of this, each step would make an iteration over the set of candidate nodes to discover the next search step's $z$.
Therefore, search step $1$ would explore $|V|-2$ nodes and the search steps $2$ to $|V|-1$ would iterate over $|V| - 1 - i$ nodes.
Summarizing this would yield:
\begin{align*}
&1 + 1\times(|V|-2) + {|V|-2 \choose 1}\times(|V|-3) + \\
&\qquad \ldots + {|V|-2 \choose |V|-2}\times0 + 1 \\
&= 1 + {|V|-2 \choose 0}\times(|V|-2) + {|V|-2 \choose 1}\times(|V|-3) + \\
&\qquad \ldots + {|V|-2 \choose |V|-2}\times0 + 1 \\
&= 2 + \sum_{i=0}^{|V|-2} {|V|-2 \choose i}\times(|V|-2-i) \\
&= 2 + (|V|-2)\times2^{|V|-3} \\
&\leq (|V|-2)\times2^{|V|-2} \qquad \text{, for $|V| \geq 4$} \\
&\le |V|\times2^{|V|}
\end{align*}
As a result, we can see that our dynamic programming-based scheduling algorithm is bounded by $\mathcal{O}(|V|\times2^{|V|})$.
By using Stirling's approximation on the complexity of the recursive topological sorting, we can prove that the dynamic programming-based scheduling algorithm should be significantly faster than the recursive topological ordering.


\end{document}